\documentclass[journal]{IEEEtran}
\ifCLASSINFOpdf
\else
\fi
\newtheorem{Theorem}{Theorem}
\newtheorem{Lemma}{Lemma}

\newtheorem{Proposition}{Proposition}
\newtheorem{Example}{Example}
\newtheorem{Remark}{Remark}
\newtheorem{Definition}{Definition}

\usepackage {amssymb}
\usepackage[numbers,sort&compress]{natbib}
\usepackage{amsmath}
\usepackage{graphicx,color}
\usepackage{subfigure,threeparttable,diagbox}
\usepackage{multirow}
\usepackage{algorithm}
\usepackage{algpseudocode}
\usepackage{graphics}
\usepackage{epsfig}
\usepackage{url,colortbl}
\newcommand{\tabincell}[2]{\begin{tabular}{@{}#1@{}}#2\end{tabular}}
\begin{document}
%
\title{Deterministic Constructions of Binary Measurement Matrices from Finite Geometry}
%
%
%

\author{Shu-Tao~Xia,~Xin-Ji~Liu,~Yong~Jiang,~and~Hai-Tao~Zheng
\thanks{This research is supported in part by the Major
State Basic Research Development Program of China (973 Program, 2012CB315803), the National Natural
Science Foundation of China (61371078, 61375054), and the Research Fund for the Doctoral Program of Higher Education of China (20130002110051).}
\thanks{All of the authors are with the Graduate School at Shenzhen, Tsinghua University, Shenzhen 518055, China (e-mail:  xiast@sz.tsinghua.edu.cn; liuxj11@mails.tsinghua.edu.cn; jiangy@sz.tsinghua.edu.cn; zheng.haitao@sz. tsinghua.edu.cn).}
}

\maketitle

\begin{abstract}
Deterministic constructions of measurement matrices in compressed sensing (CS) are considered in this paper.
The constructions are inspired by the recent discovery of Dimakis, Smarandache and Vontobel which says that parity-check matrices of good low-density parity-check (LDPC) codes can be used as {provably} good measurement matrices for compressed sensing {under $\ell_1$-minimization}.
The performance of the proposed binary measurement matrices is mainly theoretically analyzed with the help of the analyzing methods and results from (finite geometry) LDPC codes.
Particularly, several lower bounds of the spark (i.e., the smallest number of columns that are linearly dependent, which totally characterizes the recovery performance of $\ell_0$-minimization) of general binary matrices and finite geometry matrices are obtained and they improve the previously known results in most cases.
Simulation results show that the proposed matrices perform comparably to, sometimes even better than, the corresponding Gaussian random matrices.
Moreover, the proposed matrices are sparse, binary, and most of them have cyclic or quasi-cyclic structure, which will make the hardware realization convenient and easy.
\end{abstract}


%
\IEEEpeerreviewmaketitle

\section{Introduction}\label{intro}
%
%
%
%
\IEEEPARstart{C}{ompressed} sensing  (CS) \cite{ecjr, ectt, dono} is an emerging sparse sampling theory which has received large amounts of attention recently.
Consider a \emph{$k$-sparse} signal $\textit{\textbf{x}}=(x_{1}, x_{2}, \ldots, x_{n})^T\in\mathbb{R}^{n}$ with at most $k$ nonzero entries. Let $A\in \mathbb{R}^{m\times n}$ be a \emph{measurement matrix} with $m\ll n$ and $\textit{\textbf{y}}=A\textbf{\textit{x}}$ be the measurement vector.
Compressed sensing tries to recover the signal $\textit{\textbf{x}}$ from the measurement vector $\textit{\textbf{y}}$ by solving the the following \emph{$\ell_0$-minimization} problem
\begin{equation}\label{l0}
    \min ||\textit{\textbf{x}}||_{0} \quad s.t.\quad A\textbf{\textit{x}}=\textit{\textbf{y}},
\end{equation}
where $||\textit{\textbf{x}}||_{0}\triangleq |\{i: x_{i}\neq 0\}|$ denotes the $\ell_{0}$-quasi-norm of $\textit{\textbf{x}}$.
Unfortunately, it is well-known that the problem (\ref{l0}) is NP-hard in general, {cf. [MP5] in \cite{mgdj}}.
In compressed sensing, there are essentially two popular methods to deal with it.
One pursues greedy algorithms for (\ref{l0}), such as the orthogonal matching pursuit (OMP) algorithm \cite{jtag} and its modifications \cite{dnjt,wdom}.
The other one considers a convex relaxation of (\ref{l0}), or the \emph{$\ell_1$-minimization} (basis pursuit, BP) problem, as follows \cite{ectt}
\begin{equation}\label{l1}
    \min ||\textit{\textbf{x}}||_{1} \quad s.t.\quad A\textbf{\textit{x}}=\textit{\textbf{y}},
\end{equation}
where $||\textit{\textbf{x}}||_{1}\triangleq\sum_{i=1}^{n}|x_{i}|$
denotes the $\ell_{1}$-norm of $\textit{\textbf{x}}$. Note that (\ref{l1}) can be turned into a linear programming (LP) problem and thus tractable.
While considering the recovery performance, people often distinguish between the \emph{for-all} (or \emph{worst-case}) and the \emph{for-each} (or \emph{average-case}) performance \cite{sjdm,agpi}, where the former one corresponds to the situation that \textit{every} $k$-sparse signal is perfectly recovered while the latter one guarantees \emph{most}, instead of all, $k$-sparse signals are well reconstructed.

The construction of the measurement matrix $A$ is one of the main concerns in compressed sensing.
In order to select an appropriate matrix, we need some criteria.
In their earlier and fundamental work, Donoho and Elad \cite{ddme} introduced the concept of \emph{spark}. The spark of a measurement matrix $A$, denoted by ${\rm spark}(A)$, is defined to be
\begin{eqnarray}
\label{spark}
{\rm spark}(A)=\min\{||\textit{\textbf{w}}||_{0}:
\textit{\textbf{w}}\in {\rm Nullsp}_{\mathbb{R}}^{*}(A)\},
\end{eqnarray}
where
\begin{eqnarray}
\label{nullsp}
{\rm Nullsp}_{\mathbb{R}}^{*}(A)\triangleq\{\textbf{\textit{w}}\in
\mathbb{R}^n: A\textbf{\textit{w}}=\textbf{0}, \textbf{\textit{w}}\neq\textbf{0}\}.
\end{eqnarray}
Furthermore, it has been shown that if
\begin{equation}
\label{proeq1}
{\rm spark}(A)>2k,
\end{equation}
every $k$-sparse signal $\textit{\textbf{x}}$ can be exactly recovered by $\ell_0$-minimization \cite{ddme}.
In fact, it is easy to show that the condition (\ref{proeq1}) is also necessary {{for $\ell_0$-minimization}}.
Hence, spark is a relatively important performance parameter of the measurement matrix in the sense that \emph{some} signals with sparsity $k\geq{\rm spark}(A)/2$ \emph{cannot} be exactly recovered by any recovery algorithms.
Other important criteria include the coherence, restricted isometry property (RIP) \cite{ectt2} and nullspace property (NSP) \cite{wxbh,mswx}.
{It has been proved that if $A$ satisfies RIP with restricted isometry constant (RIC)
\begin{equation}\label{eq:ric}
\delta_{tk}<\sqrt{(t-1)/t}
\end{equation}
for some constant $t\geq 4/3$ \cite{cai2014}\footnote{ $\delta_{tk}<\sqrt{(t-1)/t}$ is also shown to be sharp for any $t\geq 4/3$ in \cite{cai2014}.} or the nullspace property $NSP^{<}_{\mathbb{R}}(k,C=1)$ \cite{adrs}, every $k$-sparse signal can be recovered by $\ell_1$-minimization.
For a matrix $A\in \mathbb{R}^{m\times n}$ with columns $\textit{\textbf{a}}_{1},
\textit{\textbf{a}}_{2}, \ldots, \textit{\textbf{a}}_{n}$, the \emph{coherence} of $A$ is defined as:
\begin{equation}
\label{defcoh}
  \mu(A)\triangleq \max_{1\leq i\neq j\leq n}{\frac{|\langle\textit{\textbf{a}}_{i}, \textit{\textbf{a}}_{j}\rangle|}{||\textit{\textbf{a}}_{i}||_{2}||\textit{\textbf{a}}_{j}||_{2}}},
\end{equation}
where $\langle \textit{\textbf{a}}_i, \textit{\textbf{a}}_j\rangle \triangleq \textit{\textbf{a}}_i^T \textit{\textbf{a}}_j$ and {$||\emph{\textbf{z}}||_2=\sqrt{\sum_{i=1}^m z_i^2}$ denotes the $\ell_2$-norm of $\emph{\textbf{z}}=(z_1,z_2, \ldots, z_m)^T$}.
The coherence $\mu(A)$ can be used to bound the spark and RIC of $A$ and it is shown that \cite{ddme, jbsd}:
\begin{equation}\label{generalspark}
  {\rm spark}(A)\geq 1+1/{\mu(A)},
\end{equation}
and
\begin{equation}\label{ric}
  \delta_{k}(A) \leq (k-1)\mu(A) .
\end{equation}
Therefore, any sparse signal can be exactly recovered by $\ell_0$-minimization with sparsity
\begin{equation}\label{kl0}
k_{\ell_0}<0.5+0.5/\mu(A),
\end{equation}
and by $\ell_1$-minimization with sparsity
\begin{equation}\label{kl1}
  k_{\ell_1}<0.667+0.385/\mu(A),
\end{equation}
where (\ref{kl1}) is obtained by the best condition $\delta_{\frac{3}{2}k}<\frac{\sqrt{3}}{3}$ from (\ref{eq:ric}).
However, according to the the Welch bound \cite{lwelch},
\begin{equation}\label{welch}
 \mu(A)\geq\sqrt{\frac{n-m}{m(n-1)}}.
\end{equation}
Therefore, by coherence, any matrix can only be proved to guarantee the perfect recovery of each signal with sparsity
\begin{equation}\label{sqrt-bn}
k\leq O(\sqrt{m})
\end{equation}
under both $\ell_0$-minimization and $\ell_1$-minimization, where (\ref{sqrt-bn}) is often called as the \emph{square-root bottleneck}.}
In this paper, {we will} mainly use spark to evaluate the \emph{ideal} (for-all) performance of measurement matrices since {for the proposed binary matrices, the maximum $k_{\ell_0}$ indicated by the improved spark lower bounds derived here is larger than the maximum $k_{\ell_1}$ in (\ref{kl1}) implied by coherence.
It is hoped that this will give an intuitively one-step-forward explanation on the good empirical performance of the constructed matrices.
However, it is necessary to keep in mind that this is only an intuitive and empirical statement. For general measurement matrices, large spark cannot definitely imply good practical performance since there may be some matrices with large spark but poor NSP, see the appendix for an example of such matrix.
}

Generally, constructing methods of measurement matrices can be divided into random and deterministic constructions.
Many random matrices, e.g., Gaussian matrices, partial Fourier matrices, \emph{etc.}, have been proved to satisfy RIP of order $k$ with overwhelming probability and $m= O(k\log(n/k))$ \cite{rbmd}.
However, there is no guarantee that a specific realization of random matrix works and {some} random matrices require lots of storage space. On the other hand, a deterministic matrix is often generated on the fly, and some properties, e.g., spark, coherence, RIP and NSP, could be verified definitely.
There are many works on deterministic constructions, such as \cite{lasd, sjdm, rad, sdar, aafm, lggz, wbrc, rcsh, sjthesis,aavm,slgg,Li2014a,jbsd,Yu2013}.
{Most of them are based on coherence.
For example, in \cite{rad}, DeVore construct a class of $p^2\times p^{r+1}$ matrices with coherence $\mu=r/p$, where $p$ is a prime power and $0<r<p$ is an integer; and this construction is generalized by using algebraic curves in \cite{lggz}.
Using the binary and $p$-ary BCH codes, Amini, \emph{et al.} construct the $(p^l-1)\times p^{O(p^{(l-r)\frac{\log_p r}{r}})}$ bipolar ($p=2$, \cite{aafm}) and complex ($p$ is a prime integer, \cite{aavm}) measurement matrices with coherence $\frac{p}{2(p-1)}\cdot\frac{(p^{l-r}-1)}{(p^l-1)}$, where $1<l\in \mathbb{N}$, $1\leq r\leq l-1$.
Among them, those with coherence (asymptotically) achieving the Welch bound, i.e. $k$ (asymptotically) achieving the square-root bottleneck, are of the most interest, see \cite{Li2014a,Yu2013} and references therein.
Another important class of deterministic measurement matrices is the tight frame with (nearly) optimal coherence proposed by Calderbank and his coworkers \cite{sjdm,lasd,wbrc,rcsh,sjthesis,sdar}, such as the $m\times m^2$ chirp matrices with coherence $1/\sqrt{m}$ \cite{lasd}, the $m\times m^2$ Alltop Gabor frames with coherence $1/\sqrt{m}$ \cite{wbrc} and the $2^l \times 2^{(r+2)l}$ Delsarte-Goethals (DG) frames with coherence $2^{r-l/2}$ \cite{rcsh}, where $l$ is an odd number and $0\leq r\leq (l-1)/2$ is a constant integer.
Apart from the for-all performance (\ref{kl1}) under $\ell_1$-minimization guaranteed by coherence, the for-each performance under $\ell_1$-minimization of these tight frames is analyzed through the \emph{statistical RIP} (StRIP) \cite{rcsh} and \cite[Th. 1.2]{candes2009}.
In particular, if the chirp matrices, Alltop Gabor frames or DG frames are taken as measurement matrices, the $k$-sparse signal $\emph{\textbf{x}}$ with uniformly random support, uniformly random sign (for nonzero entries) and sparsity $k\leq O(m/\log n)$ can obtain perfect recovery under $\ell_1$-minimization with probability $1-O(1/n)$ \cite{sjdm,sjthesis}.
In the following, we usually use $A$ to denote a real matrix and $H$ a binary matrix.

Recently, connections between low-density parity-check (LDPC) codes \cite{gallagerit} and compressed sensing excite interests. Dimakis, Smarandache, and Vontobel \cite{adrs} point out that the LP decoding of LDPC codes is very similar to the LP reconstruction (i.e., $\ell_1$-minimization) of CS, and further show that parity-check matrices of good LDPC codes can be used as provably good measurement matrices under $\ell_1$-minimization.
LDPC codes are a class of linear block codes, each of which is defined by the nullspace over $\mathbb{F}_2=\{0,1\}$ of a
binary sparse $m\times n$ parity-check matrix $H$.
$H$ is said to be $(\gamma,\rho)$-\emph{regular} if $H$ has the uniform column weight $\gamma$ and the uniform row weight $\rho$.
{In \cite{wlkk,akat,atad}, the famous progressive edge-growth (PEG) algorithm \cite{xheeit} for LDPC codes is used to construct binary sparse measurement matrices. These matrices show empirically \cite{wlkk} and provably (\cite{akat,atad}, under $\ell_1$-minimization) good performance in CS.}

Inspired by the connection between LDPC codes and CS in \cite{adrs}, we construct the deterministic measurement matrices from finite geometry LDPC (FG-LDPC) codes.
Our main contributions focus on the following two aspects.
\begin{itemize}
\item
\emph{Constructing two classes of deterministic measurement matrices from finite geometry.} LDPC codes based on finite geometry (FG) could be found in \cite{klf, txla05}. With similar methods, two classes of deterministic measurement matrices based on finite geometry are given.
Numerous experiments are presented and show that the proposed matrices perform empirically as well as, sometimes better than, the corresponding Gaussian matrices under both {BP and OMP}, even for the noisy situations.
Moreover, most of the proposed matrices could be put in either cyclic or quasi-cyclic form, thus making the hardware realization of sampling easier and simpler.
\item
\emph{Lower bounding the spark of a binary measurement matrix $H$.}
The spark of a measurement matrix is useful since it totally characterize the for-all performance of $\ell_0$-minimization.
We firstly obtain a new lower bound of ${\rm spark}(H)$ for general binary matrices, which improve the traditional one (\ref{generalspark}) in most cases. Afterwards, for the first class of binary matrices from finite geometry, we give two further improved lower bounds to show their relatively large spark. The fact that the proposed matrices have relatively large spark can explain \emph{to some extent} their empirically good performance under both BP and OMP.
\end{itemize}

{ After the submission of this paper, we realized that similar constructions via finite geometry are also proposed by Li and Ge \cite{slgg}. However, our work has been carried out independently and concurrently and differs from \cite{slgg} in three aspects.
\begin{itemize}
  \item All incidence matrices of $\mu_2$-flat over $\mu_1$-flat in finite geometry are covered in this paper, while only the line-point incidence matrices (i.e. $\mu_1=0$, $\mu_2=1$) are considered in \cite{slgg}, see Section \ref{subsec:classImat}.
  \item The parallel structure of Euclidean geometry is utilized here to obtain measurement matrices with a bit more diversified sizes and this is not included in \cite{slgg}, see Section \ref{sec:parallel}.
  \item The binary characteristic of the proposed matrices is used to get better spark bounds, while in \cite{slgg}, only coherence plays directly as the key tool for performance analysis and the binary characteristic is ignored, see Section \ref{mainresults}.
\end{itemize}
}

The rest of this paper is organized as follows.
Section \ref{fgldpc} gives a brief introduction to finite geometries and their parallel and quasi-cyclic structures, which result in the two classes of deterministic constructions.
Section \ref{mainresults} gives a lower bound of spark for general binary matrices and two further improved lower bounds for the proposed matrices from finite geometry.
Lots {{of}} simulations are given in Section \ref{simulation}.
Section \ref{conclusion} concludes the paper with some discussions.

\section{Measurement Matrices from Finite Geometries}
\label{fgldpc}
Finite geometry was used to construct several classes of parity-check matrices of LDPC codes which manifest excellent performance under iterative decoding \cite{klf} \cite{txla05}.
In the following, we briefly introduce some notations and results of finite geometry \cite{txla05}\cite[pp. 692-702]{fmns}.

Let $\mathbb{F}_q$ be a finite field with $q$ elements and
$\mathbb{F}_q^r$ be the $r$-dimensional vector space over
$\mathbb{F}_q$, where $r\ge 2$.
Let $EG(r,q)$ be the $r$-dimensional Euclidean geometry over
$\mathbb{F}_q$. $EG(r,q)$ has $q^r$ points,
which are vectors of $\mathbb{F}_q^r$. The $\mu$-flat in $EG(r,q)$
is a $\mu$-dimensional subspace of $\mathbb{F}_q^r$ or its coset.
Let $PG(r,q)$ be the $r$-dimensional projective geometry over
$\mathbb{F}_q$. $PG(r,q)$ is defined in
$\mathbb{F}_q^{r+1}\setminus\{\mathbf{0}\}$. Two nonzero vectors
$\mathbf{p,p'}\in \mathbb{F}_q^{r+1}$ are said to be equivalent if
there is $\lambda\in \mathbb{F}_q$ such that $\mathbf{p}=\lambda
\mathbf{p'}$. It is well known that all equivalence classes of
$\mathbb{F}_q^{r+1}\setminus\{\mathbf{0}\}$ form points of
$PG(r,q)$. $PG(r,q)$ has $(q^{r+1}-1)/(q-1)$ points. The
$\mu$-flat in $PG(r,q)$ is simply the set of equivalence classes contained
in a $({\mu}+1)$-dimensional subspace of $\mathbb{F}_q^{r+1}$.
In this paper, in order to present a unified approach, we use
$FG(r,q)$ to denote either $EG(r,q)$ or $PG(r,q)$.
A \emph{point} is a $0$-flat and a \emph{line} is a $1$-flat.

\subsection{Incidence Matrix in Finite Geometry}\label{subsec:classImat}
For $0\le \mu_1<\mu_2\le r$ \cite{txla05}, there are $N(\mu_2,\mu_1)$
$\;\mu_1$-flats contained in a given $\mu_2$-flat and
$A(\mu_2,\mu_1)$ $\;\mu_2$-flats containing a given $\mu_1$-flat,
where for $EG(r,q)$ and $PG(r,q)$ respectively
\begin{eqnarray}
\label{eq:fg1}
N_{EG}(\mu_2,\mu_1)&=& q^{\mu_2-\mu_1} \prod_{i=1}^{\mu_1} \frac{q^{\mu_2-i+1}-1}{q^{\mu_1-i+1}-1},
\end{eqnarray}
\begin{eqnarray}
\label{eq:fg2}
N_{PG}(\mu_2,\mu_1)&=& \prod_{i=0}^{\mu_1} \frac{q^{\mu_2-i+1}-1}{q^{\mu_1-i+1}-1},
\end{eqnarray}
\begin{eqnarray}
\label{eq:fg3} A_{EG}(\mu_2,\mu_1)=A_{PG}(\mu_2,\mu_1)=
\prod_{i=\mu_1+1}^{\mu_2} \frac{q^{r-i+1}-1}{q^{\mu_2-i+1}-1}.
\end{eqnarray}
Let $n=N(r,\mu_1)$ and $J=N(r,\mu_2)$ be the numbers of
$\mu_1$-flats and $\mu_2$-flats in $FG(r,q)$ respectively. The
$\mu_1$-flats and $\mu_2$-flats are indexed from $1$ to $n$ and $1$
to $J$ respectively. The \emph{incidence matrix $H=(h_{ji})$ of $\mu_2$-flat over $\mu_1$-flat} is a binary $J\times n$
matrix, where $h_{ji}=1$ for $1\le j\le J$ and $1\le i \le n$ if and
only if the $j$th $\mu_2$-flat contains the $i$th $\mu_1$-flat. The
rows of $H$ correspond to all the $\mu_2$-flats in $FG(r,q)$ and the columns of $H$ correspond
to all the $\mu_1$-flats in $FG(r,q)$.
Moreover, $H$ is a $(\gamma,\rho)$-regular matrix, where
\begin{eqnarray}
\gamma=A(\mu_2,\mu_1),\quad\rho=N(\mu_2,\mu_1).
\end{eqnarray}
{Since an $m\times n$ measurement matrix in CS should satisfy $m<n$, we construct the \emph{class-I finite geometry measurement matrix} as follows.
\begin{itemize}
  \item If $J<n$, use $H$ directly as the measurement matrix, and $H$ is called the \emph{type-I} (in class-I) \emph{finite geometry measurement matrix}.
  \item If $J>n$, use $H^T$ directly as the measurement matrix, and $H^T$ is called the \emph{type-II} (in class-I) \emph{finite geometry measurement matrix}.
\end{itemize}

Using the properties of finite geometry, it is easy to find that the inner product of two different columns of $H$ equals to the number of $\mu_{2}$-flats containing two fixed $\mu_{1}$-flats simultaneously, whose maximum value is $A(\mu_2,\mu_1+1)$; while the inner product of two different rows of $H$ equals to the number of $\mu_{1}$-flats contained by two fixed $\mu_2$-flats simultaneously, whose maximum value is $N(\mu_2-1,\mu_1)$.
Therefore, we have the following result.
\begin{Proposition}\label{pro:cohfg}
Let $r, \mu_1,\mu_2$ be integers, $0\le
\mu_1<\mu_2<r$, $H$ be the type-I finite geometry measurement matrix and $H^T$ be the type-II finite geometry measurement matrix. Then
  \begin{eqnarray}
    \mu(H) = \frac{A(\mu_{2}, \mu_{1}+1)}{A(\mu_{2}, \mu_{1})},\\
    \mu(H^T) = \frac{N(\mu_2-1,\mu_1)}{N(\mu_2,\mu_1)}.
  \end{eqnarray}
\end{Proposition}

}

For any pair $(J,n)$, whether $J>n$, $J<n$ or $J=n$, we could construct another large class of measurement matrices with a bit more diversified sizes by removing some rows and columns from $H$ or $H^T$ \emph{in a deterministic way}, and we call them \emph{the class-II finite geometry measurement matrices}.
In this paper, an efficient method to remove rows and columns deterministically from the class-I finite geometry measurement matrices is proposed and it makes use of the the parallel structure in Euclidean Geometry.
\subsection{Measurement Matrices from the Parallel Structure in Euclidean Geometry}
\label{sec:parallel}
Since a projective geometry does not have the parallel structure, we concentrate on $EG(r,q)$ only. In $EG(r,q)$, a $\mu$-flat contains $q^\mu$ points and two $\mu$-flats are either disjoint or intersecting on a flat with dimension at most $\mu-1$. The $\mu$-flats that
correspond to the cosets of a $\mu$-dimensional subspace of $\mathbb{F}_q^r$ (including the subspace itself) are said to be \emph{parallel} to each other
and form a \emph{parallel $\mu$-flat bundle}. The $\mu$-flats within a parallel $\mu$-flat bundle are disjoint and contain all the points of $EG(r,q)$ with each point appearing
once and only once. The number of $\mu$-flats in a parallel $\mu$-flat bundle is $q^{r-\mu}$.

There are totally $J=N(r,\mu_2)$ $\mu_{2}$-flats which consist of  $K=J /q^{r-\mu_{2}}$ parallel $\mu_2$-flat bundles in $EG(r,q)$. We index these parallel bundles from $1$ to $K$. Consider the $J\times n$ incidence matrix $H$ of $\mu_{2}$-flat over $\mu_{1}$-flat. All $J$ rows of $H$ could be divided into $K$ bundles each of which contains $q^{r-\mu_2}$ rows, i.e., by suitable row arrangement, $H$ could be written as
\begin{equation}
\label{Hpara}
H=(H_1^T,H_2^T,\ldots,H_{K}^T)^{T},
\end{equation}
where $H_i$ ($1\le i\le K$) is a $q^{r-\mu_{2}}\times n$ submatrix of $H$ corresponding to the $i$-th parallel $\mu_2$-flat bundle.
Similarly, the columns of $H$, which correspond to all the $n$ $\mu_1$-flats, can also be ordered according to the parallel $\mu_1$-flat bundles in $EG(r, q)$. By deleting some rows or columns corresponding to the parallel bundles from $H$, and transposing the obtained submatrix if needed, we could construct a large amount of measurement matrices with various sizes.

In the following, using the \emph{Euclidean plane $EG(2,q)$}, we give a detailed example to show the above method to remove rows and columns from a class-I Euclidean geometry matrix.
$EG(2,q)$ has $n = q^2$ points, $J = (q^2+q)$ lines and a $(q^2+q)\times q^2$ line-point incidence matrix $H$.
All the $J$ lines can be divided into $(q+1)$ parallel line bundles each of which consists of $q$ lines. All the $n$ points are parallel to each other and form a trivial parallel point bundle with $n$ points. By (\ref{Hpara}), $H$ can be arranged as
\begin{equation}\label{eq:EGmat}
  H=(H_1^T, H_2^T, \ldots, H_{q+1}^T)^{T},
\end{equation}
where for $i=1,\ldots,q+1$, the $q$ rows of $H_i$ correspond to the $q$ lines in the $i$-th parallel line bundle.

Next, we will remove some rows and columns of $H$ according to the parallel structure of the lines in $EG(2,q)$.
Firstly, by simply choosing the first $\gamma$ submatrices $H_i$, $1\leq i\leq\gamma$ and $\gamma>1$, we can construct a $\gamma q\times q^2$ matrix
\begin{equation}
{H(\gamma,q,q)}=(H_1^T, H_2^T, \ldots, H_{\gamma}^T)^{T}.
\end{equation}
Since every point occurs exactly once in each parallel line bundle and every line contains exactly $q$ points, $H(\gamma,q,q)$ is $(\gamma,q)$-regular.
In the following, we delete some columns of $H(\gamma,q,q)$ in the way such that it keeps the regularity of the resulting matrix.
Recall that for any fixed $H_{i}$, its corresponding $q$ lines are parallel to each other and partition the geometry.
Let $\rho$ be an integer with $\gamma\le\rho\leq q$. Select the first $(q - \rho)$ lines from the $(\gamma+1)$-th parallel line bundle which contain exactly $q(q - \rho)$ points.
Remove the $q(q - \rho)$ columns of $H(\gamma,q,q)$ corresponding to the $q(q - \rho)$ points lying on the selected $(q-\rho)$ lines, and then we can obtain a $\gamma q\times \rho q$ submatrix $H(\gamma,\rho,q)$, i.e., the class--II finite geometry measurement matrix.
It is easy to see that the $q$ points on the same line of the $(\gamma+1)$-th parallel line bundle lie on exactly $q$ lines of each of the first $\gamma$ parallel line bundles.
Therefore, $H(\gamma,\rho,q)$ is $(\gamma,\rho)$-regular.
Considered that every two points in $EG(2,q)$ lie on exactly one line, the maximum inner product of any two different columns of $H(\gamma,\rho,q)$ will equal 1 as long as $\rho\ge2$.
Hence, we have the following result.
\begin{Proposition}
  Let $H(\gamma,\rho,q)$ be the $\gamma q\times \rho q$ class--II finite geometry measurement matrix from $EG(2,q)$, $1<\gamma\le \rho\le q$,
  \begin{equation}\label{eq:cohHsub}
  \mu(H(\gamma,\rho,q))=\frac{1}{\gamma}.
  \end{equation}
\end{Proposition}

\begin{Remark}\label{kscaling}
Let $\gamma = cq$ and $\rho=q$, where $0<c\le1$ is a constant, then $H(\gamma,\rho,q)$ will be a binary matrix with $m=cn$ rows, $n=q^2$ columns, and coherence $\mu(H(\gamma,\rho,q))=\frac{1}{\gamma}=\sqrt{\frac{1}{cm}}$.
Therefore, according to (\ref{kl0}) and (\ref{kl1}), signals measured by $H(\gamma,\rho,q)$ and with sparsity $k_{\ell_0}/k_{\ell_1}=O(\sqrt{m})$ can be exactly recovered by both $\ell_0$-optimization and $\ell_1$-optimization, where $k_{\ell_0}$ and $k_{\ell_1}$ almost approach the square-root bottleneck, although with a gap of constant.
\end{Remark}
\begin{Remark}
  Suppose we want to construct an $m\times n$ class--II finite geometry measurement matrix, where $m$ and $n$ can be written as $m=\gamma_1 q_1 = \gamma_2 q_2$ and $n= \rho_1 q_1 = \rho_2 q_2$ at the same time. Suppose $q_1<q_2$, then we should choose $H(\gamma_1,\rho_1,q_1)$ since it will have lower coherence according to (\ref{eq:cohHsub}) and larger spark lower bound according to Theorem \ref{preth} in Section \ref{sec:generalbound}.
\end{Remark}

\subsection{Cyclic and Quasi-cyclic Structure in Finite Geometries}
Apart from the parallel structure of Euclidean geometry, most of the incidence matrices in Euclidean geometry and projective geometry also have cyclic or quasi-cyclic structure \cite{txla05}. This is accomplished
by grouping the flats of two different dimensions of
a finite geometry into cyclic classes. For a Euclidean geometry, only the flats not
passing through the origin are used for matrix construction. Based on this grouping of rows and columns, the incidence matrix in finite geometry consists of square submatrices (or blocks), and each of these square submatrices is a circulant matrix in which each row is a cyclic shift of the row above it and the first row is the cyclic shift of the last row.
Moreover, this skill is also compatible with the parallel structure of Euclidean geometry. Hence, the sampling process with these measurement
matrices is easy and can be achieved with linear shift registers. For more details, please refer to \cite[Appendix A]{txla05}.

\section {Spark Analysis of Binary Matrices}
\label{mainresults}
As has been stated in (\ref{proeq1}), ${\rm spark}(A)>2k$ is the necessary and sufficient condition for $\ell_0$-minimization to perfectly recover any $k$-sparse signal.
While choosing measurement matrices, those with large sparks are intuitively preferred as good candidates.
However, the computation of spark is generally NP-hard \cite{mgdj}.
In this section, we give several new lower bounds of the spark for general binary matrices and finite geometry matrices.
The relatively large spark lower bounds of the proposed matrices {can explain their empirical good performance in Section \ref{simulation} to some extent}.

\subsection{Lower Bound of Spark for General Binary Matrices}\label{sec:generalbound}
For a real vector $\textit{\textbf{x}}\in\mathbb{R}^{n}$, the support of $\textit{\textbf{x}}$ is defined by the set of non-zero positions, i.e.,
${\rm supp}(\textit{\textbf{x}})\triangleq \{i: x_{i}\neq 0\}$.

Consider a binary $m\times n$ matrix $H$ with minimum column weight $\gamma>0$. Suppose the maximum inner product of any two different columns of $H$ is $\lambda>0$.
By (\ref{defcoh}), we have $\mu(H)\leq\frac{\lambda}{\gamma}.$
According to the lower bound (\ref{generalspark}) from \cite{ddme},
\begin{equation}
\label{binarysparkold}
 {\rm spark}(H)\geq 1+\frac{\gamma}{\lambda}.
\end{equation}
{In addition, from (\ref{kl0}) and (\ref{kl1}), any signal with sparsity
\begin{equation}\label{kl0binary}
  k_{\ell_0}<0.5+\frac{0.5\gamma}{\lambda}
\end{equation}
and
\begin{equation}\label{kl1binary}
  k_{\ell_1}<0.667+\frac{0.385\gamma}{\lambda}
\end{equation}
can get perfect recovery under $\ell_0$-minimization and $\ell_1$-minimization, respectively.}

As a matter of fact, for the general binary matrix $H$, we often have a tighter lower bound of its spark.
\begin{Theorem}
\label {preth}
Let $H$ be a binary $m\times n$ matrix with minimum column weight $\gamma>0$, and suppose the maximum inner product of any two different columns of $H$ is $\lambda>0$. Then
\begin{equation}
\label{thgeneral}
  {\rm spark}(H)\geq \frac{2\gamma}{\lambda}.
\end{equation}
\end{Theorem}

\begin{proof}
For any $\textit{\textbf{w}}=(w_{1}, w_{2}, \ldots, w_{n})\in {\rm Nullsp}_{\mathbb{R}}^{*}(H)$, we split
the non-empty set $\rm {\rm supp}(\textit{\textbf{w}})$ into two parts $\rm {\rm supp}(\textit{\textbf{w}}^{+})$ and $\rm {\rm supp}(\textit{\textbf{w}}^{-})$,
\begin{eqnarray}
\label{split1}
\rm {\rm supp}(\textit{\textbf{w}}^{+})&\triangleq&\{i: w_{i}>0\},\\
\label{split2}
\rm {\rm supp}(\textit{\textbf{w}}^{-})&\triangleq&\{i: w_{i}<0\}.
\end{eqnarray}
Without loss of generality, we assume that $|\rm {\rm supp}(\textit{\textbf{w}}^{+})|\geq|\rm {\rm supp}(\textit{\textbf{w}}^{-})|$.
For fixed $j\in \rm {\rm supp}(\textit{\textbf{w}}^{+})$, by selecting the $j$-th column of $H$ and all
the columns in $\rm {\rm supp}(\textit{\textbf{w}}^{-})$ of $H$, we get a submatrix $H(j)$.
Since the column weight of $H$ is at least $\gamma$, we could select $\gamma$ rows of $H(j)$ to form a $\gamma\times(1+|\rm {\rm supp}(\textit{\textbf{w}}^{-})|)$ submatrix of $H$, say $H(\gamma,j)$, where the column corresponds to $j$ is all 1 column. Now let's count the total number of 1's of $H(\gamma,j)$ in two ways.
\begin{itemize}
\item
From the view of columns, since the maximum inner product of any two different columns of $H$ is $\lambda$, each of the columns of $H(\gamma,j)$ corresponds to $\rm {\rm supp}(\textit{\textbf{w}}^{-})$ has at most $\lambda$ 1's. So the total number is at most $\gamma+\lambda|\rm {\rm supp}(\textit{\textbf{w}}^{-})|$.
\item
From the view of rows, we claim that there is at least two 1's in each row of $H(\gamma,j)$, which implies the total number is at least $2\gamma$ 1's. The claim is shown as follows.
Let $\textit{\textbf{h}}(j)$ be any row of $H(\gamma,j)$ and $\textit{\textbf{h}}=(h_1,\ldots,h_n)$ be its corresponding row in $H$. Note that $h_j=1$. Since $\textit{\textbf{w}}\in {\rm Nullsp}_{\mathbb{R}}^{*}(H)$,
\begin{eqnarray*}
0=\sum_{i\in \rm {\rm supp}(\textit{\textbf{w}})}w_i h_i=\sum_{i\in \rm supp(\textit{\textbf{w}}^+)}w_i h_i+\sum_{i\in \rm supp(\textit{\textbf{w}}^-)}w_i h_i,
\end{eqnarray*}
which implies that
\begin{eqnarray*}
-\sum_{i\in \rm supp(\textit{\textbf{w}}^-)}w_i h_i=\sum_{i\in \rm supp(\textit{\textbf{w}}^+)}w_i h_i\ge w_j>0.
\end{eqnarray*}
So there are at least one 1's in $\{h_i : i\in \rm supp(\textit{\textbf{w}}^-)\}$ and $\textit{\textbf{h}}(j)$ has at least two 1's.
\end{itemize}
Therefore, $2\gamma\leq \gamma+\lambda|\rm {\rm supp}(\textit{\textbf{w}}^{-})|$, which implies that $|\rm {\rm supp}(\textit{\textbf{w}}^{-})|\geq\frac{\gamma}{\lambda}$.  Since $|\rm {\rm supp}(\textit{\textbf{w}}^{+})|\geq|\rm {\rm supp}(\textit{\textbf{w}}^{-})| \geq\frac{\gamma}{\lambda}$, $$|\rm {\rm supp}(\textit{\textbf{w}})|=|\rm {\rm supp}(\textit{\textbf{w}}^{+})|+|\rm {\rm supp}(\textit{\textbf{w}}^{-})| \geq \frac{2\gamma}{\lambda}$$ and the conclusion (\ref{thgeneral}) follows.
\end{proof}
\begin{Remark}\label{rem:extend}
Obviously, the lower bound (\ref{thgeneral}) is tighter than (\ref{binarysparkold}).
Combining (\ref{thgeneral}) with (\ref{proeq1}), we have that any signal with sparsity
\begin{equation}\label{kl0new}
  k_{\ell_0}<\frac{\gamma}{\lambda}.
\end{equation}
can be exactly recovered by $\ell_0$-minimization, which improves (\ref{kl0binary}) by a factor of about 2.
{Particularly, if $H$ has \emph{uniform} column weight $\gamma$, $\mu(H)=\frac{\lambda}{\gamma}$ and (\ref{thgeneral}) is equivalent to
\begin{equation}
\label{thgeneralmu}
  {\rm spark}(H)\geq \frac{2}{\mu(H)},
\end{equation}
which improves (\ref{generalspark}).
Moreover, in a subsequent work \cite{xlsx1}, we further prove that all signals with sparsity
\begin{equation}\label{kl1new}
  k_{\ell_1}<\frac{\gamma}{\lambda}=\frac{1}{\mu(H)}
\end{equation} can obtain perfect recovery under $\ell_1$-minimization, which is better than (\ref{kl1binary}) and (\ref{kl1}) implied by the coherence}.
{Therefore, any signal measured by the class--II finite geometry measurement matrix $H(\gamma,\rho,q)$ and with sparsity
\begin{equation}
  k_{\ell_0} = k_{\ell_1} < \gamma
\end{equation}
can be perfectly recovered by both $\ell_0$-minimization and $\ell_1$-minimization.}
Finally, the coincidence of (\ref{kl0new}) and (\ref{kl1new}) reveals that perhaps in some cases, the results of spark and $\ell_0$-minimization could be strengthened to the corresponding results of $\ell_1$-minimization.
\end{Remark}
\begin{Remark}
Consider a complete graph on $4$ vertices with point-line incidence matrix:
\begin{eqnarray*}
H = \left(
\begin{array}{cccccc}
1&1&1&0&0&0 \\
1&0&0&1&1&0 \\
0&1&0&1&0&1 \\
0&0&1&0&1&1
\end{array}
\right).
\end{eqnarray*}
Clearly, $\gamma=2,\lambda=1$ and ${\rm spark}(H)\geq 4$ according to (\ref{thgeneral}). Moreover, since $(1,-1,0,0,-1,1)\in {\rm Nullsp}_{\mathbb{R}}^{*}(H)$, ${\rm spark}(H)=4$, which means that the lower bound (\ref{thgeneral}) could be achieved.
\end{Remark}

\subsection{Lower Bounds of Spark for the class--I Finite Geometry Matrices}
Theorem \ref{preth} does apply to all matrices constructed in Section \ref{fgldpc} from finite geometry. In this part, {we will} show that for the class--I finite geometry measurement matrices, better spark lower bounds could be obtained.

Let $H$ be the type-I $J\times n$ incidence matrix of $\mu_{2}$-flat over $\mu_{1}$-flat in $FG(r, q)$, where $0\leq\mu_{1}<\mu_{2}<r$, $n=N(r, \mu_{1})$ and $J=N(r, \mu_{2})$.
\begin{Lemma} \cite{sxff}\quad
\label{lem1}
Let $0\le \mu_1<\mu_2< r$ and $1\le l\le
A(\mu_2,\mu_2-1)$. Given any $l$ different $\mu_1$-flats $\mathcal{F}_1, \mathcal{F}_2, \ldots, \mathcal{F}_l$ in $FG(r,q)$ and for any $1\le j\le l$, there exists one $(\mu_2-1)$-flat $\mathcal{F}$ such that $\mathcal{F}_j \subseteq \mathcal{F}$ and $\mathcal{F}_i\not\subseteq \mathcal{F}$ for all $i=1,\ldots, j-1, j+1,\ldots l$.
\end{Lemma}

\begin{Theorem}
\label{thm1} Let $r, \mu_1,\mu_2$ be integers, $0\le
\mu_1<\mu_2<r$ and $H$ be the type-I finite geometry measurement matrix. Then
\begin{eqnarray}
{{\rm spark}(H)}\ge 2A(\mu_2,\mu_2-1),
\label{eqthm1}
\end{eqnarray}
where
\begin{eqnarray*}
A(\mu_2,\mu_2-1)&=&\frac{q^{r-\mu_2+1}-1}{q-1}.
\end{eqnarray*}
\end{Theorem}

\begin{proof}
Let $$u=A(\mu_2,\mu_2-1)$$ and assume the contrary that $${\rm spark}(H)<2u.$$
Select a $\textit{\textbf{w}}=(w_{1}, w_{2}, \ldots, w_{n})\in {\rm Nullsp}_{\mathbb{R}}^{*}(H)$ such that $|{\rm supp}(\textit{\textbf{w}})|={\rm spark}(H)$. By (\ref{split1}) and (\ref{split2}), we split
the non-empty set ${\rm supp}(\textit{\textbf{w}})$ into two parts ${\rm supp}(\textit{\textbf{w}}^{+})$ and ${\rm supp}(\textit{\textbf{w}}^{-})$, and
assume $|{\rm supp}(\textit{\textbf{w}}^{+})|\geq|{\rm supp}(\textit{\textbf{w}}^{-})|$ without loss of generality. Thus by the assumption $$|{\rm supp}(\textit{\textbf{w}}^{-})|<u \quad\mbox{or}\quad |{\rm supp}(\textit{\textbf{w}}^{-})|\le u-1.$$
For fixed $j\in {\rm supp}(\textit{\textbf{w}}^{+})$, by selecting $j$-th column of $H$ and all the columns in ${\rm supp}(\textit{\textbf{w}}^{-})$ of $H$, we get a submatrix $H(j)$. The number of columns in $H(j)$ is $1+|{\rm supp}(\textit{\textbf{w}}^{-})|$ and not greater than $u$. Let $\mathcal{F}_j$ and $\{\mathcal{F}_i, i\in {\rm supp}(\textit{\textbf{w}}^{-})\}$ be the $\mu_1$-flats corresponding to the columns of $H(j)$. By Lemma \ref{lem1}, there exists one $(\mu_2-1)$-flat $\mathcal{F}$ such that $\mathcal{F}_j \subseteq \mathcal{F}$ and $\mathcal{F}_i\not\subseteq \mathcal{F}$ for all $i\in {\rm supp}(\textit{\textbf{w}}^{-})$. There are exactly $u$ $\;\mu_2$-flats containing $\mathcal{F}$. Note that among these $\mu_2$-flats, any two distinct $\mu_2$-flats have no other common points except those
points in $\mathcal{F}$ (see \cite{txla05}). Hence, each of these $u$ $\;\mu_2$-flats contains the $\mu_1$-flat $\mathcal{F}_j$ and for
any $i\in {\rm supp}(\textit{\textbf{w}}^{-})$, there exist at most one of these $u$
$\;\mu_2$-flats containing the $\mu_1$-flat $\mathcal{F}_i$. In other words,
there exist $u$ rows in $H(j)$ such that each of
these rows has component $1$ at position $j$ and for any $i \in {\rm supp}(\textit{\textbf{w}}^{-})$,
there exists at most one row that has component $1$ at position $i$.

Let $H(u,j)$ be the $u \times (1+|{\rm supp}(\textit{\textbf{w}}^{-})|)$ submatrix of $H(j)$ by choosing these rows, where the column corresponds to $j$ is all 1 column. Now let's count the total number of 1's of $H(u,j)$ in two ways.
The column corresponds to $j$ has $u$
$1$'s while each of the other columns has at most
one $1$. Thus from the view of columns, the total number of $1$'s in
$H(u,j)$ is at most $u+|{\rm supp}(\textit{\textbf{w}}^{-})|$. On the other hand, suppose $x$ is the number of rows in $H(u,j)$ with weight one. Then,
there are $u-x$ rows with weight at least two. Thus
from the view of rows, the total number of $1$'s in $H(u,j)$ is at
least $x+2(u-x)$. Hence, $x+2(u-x)\le u+|{\rm supp}(\textit{\textbf{w}}^{-})|$, which implies that $x\ge
u-|{\rm supp}(\textit{\textbf{w}}^{-})| \ge 1$ by the assumption. In other words, $H(j)$ contains a row with value $1$ at the position corresponding to $j$ and $0$ at other positions.
Denote this row by $\textit{\textbf{h}}(j)$ and let $\textit{\textbf{h}}=(h_1,\ldots,h_n)$ be its corresponding row in $H$. Note that $h_j=1$ and $h_i=0, i \in {\rm supp}(\textit{\textbf{w}}^{-})$. Since $\textit{\textbf{w}}\in {\rm Nullsp}_{\mathbb{R}}^{*}(H)$,
\begin{eqnarray*}
0&=&\sum_{i\in {\rm supp}(\textit{\textbf{w}})}w_i h_i=\sum_{i\in supp(\textit{\textbf{w}}^+)}w_i h_i+\sum_{i\in supp(\textit{\textbf{w}}^-)}w_i h_i\\
&=& \sum_{i\in supp(\textit{\textbf{w}}^+)}w_i h_i \ge w_j >0,
\end{eqnarray*}
which leads to a contradiction. Therefore, the assumption is wrong and the theorem follows by (\ref{eq:fg3}).
\end{proof}
\begin{Remark}
Combining Theorem \ref{thm1} with (\ref{proeq1}), we have that when the type-I finite geometry measurement matrix $H$ is used, any sparse signal with sparsity{
\begin{equation}\label{kl0new1}
k_{\ell_0}<A(\mu_2,\mu_2-1)
\end{equation}}
can be exactly recovered by $\ell_0$-minimization.
\end{Remark}
\begin{Remark}
\label{remark2}
 By Theorem \ref{preth},
\begin{eqnarray}
{\rm spark}(H)\ge \frac{2A(\mu_2,\mu_1)}{A(\mu_2,\mu_1+1)}\label{bound11}.
\end{eqnarray}
It is easy to verify by (\ref{eq:fg3}) that
the lower bound (\ref{eqthm1}) is tighter than (\ref{bound11}) when $\mu_2>\mu_1+1$.
\end{Remark}

\begin{Remark}\label{rem:type1}
  Suppose $q$ is large, according to (\ref{eq:fg1})--(\ref{eq:fg3}),
  \begin{eqnarray}
    N_{EG}(\mu_2,\mu_1) &\approx& N_{PG}(\mu_2,\mu_1)\approx q^{(\mu_2-\mu_1)\cdot(\mu_1+1)},\nonumber\\
    A_{EG}(\mu_2,\mu_1) &=& A_{PG}(\mu_2,\mu_1)\approx q^{(r-\mu_2)\cdot(\mu_2-\mu_1)}.\nonumber
  \end{eqnarray}
  As a result, $H$ has $m=N(r,\mu_2)\approx q^{(r-\mu_2)\cdot(\mu_2+1)}$ rows, $n=N(r,\mu_1)\approx q^{(r-\mu_1)\cdot(\mu_1+1)}$ columns and ${\rm spark}(H)\ge 2A(\mu_2,\mu_2-1)\approx 2q^{(r-\mu_2)}\approx 2\cdot\sqrt[\mu_2+1]{m}$, which means that any $k$-sparse signal can be exactly recovered by $\ell_0$-minimization with $k=O(\sqrt[\mu_2+1]{m})\leq O(\sqrt{m})$.
  Such matrix can not be proved to overcome the square-root bottleneck in the sense of $\ell_0$-minimization.
\end{Remark}

Similarly, for the type-II finite geometry measurement matrix, we have the next result by \cite[Lemma 2]{sxff}.

\begin{Theorem}
\label {thm2}
Let $r, \mu_1,\mu_2$ be integers, $0\le
\mu_1<\mu_2<r$ and $H^T$ be the type-II finite geometry measurement matrix. Then
\begin{eqnarray}
{\rm spark}(H^T)\ge 2N(\mu_1+1,\mu_1),
\label{eqthm2}
\end{eqnarray}
where for Euclidean geometry (EG) and projective geometry (PG) respectively
\begin{eqnarray*}
N_{EG}(\mu_1\!+\!1,\mu_1)\!=\frac{q^{\mu_1+2}\!\!-\!q}{q-1},\;
N_{PG}(\mu_1\!+\!1,\mu_1)\!=\frac{q^{\mu_1+2}\!\!-\!1}{q-1}.
\end{eqnarray*}
\end{Theorem}

Finally, we summarize the parameters and available performance guarantees under both $\ell_0$-minimization and $\ell_1$-minimization of the binary measurement matrices proposed in this paper in Table \ref{tab:parameter}.
\begin{table*}[htbp]
\centering
\begin{threeparttable}[b]
\renewcommand{\arraystretch}{1.2}
   \caption{The Parameters and Theoretical Performance of Binary Measurement Matrices from Finite Geometry}
   \label{tab:parameter}
   \begin{tabular}{|c|c|c|c|c|c|c|c|}
     \hline matrix&$m$&$n$&spark&$\mu$&$k_{\ell_0}$&$k_{\ell_1}$&parameter constraints\\
     \hline\hline \tabincell{c}{class--I,\\type-I}&$N(r, \mu_2)$&$N(r, \mu_1)$&$\geq2A(\mu_2,\mu_2-1)$&$\frac{A(\mu_2,\mu_1+1)}{A(\mu_2,\mu_1)}$&$<A(\mu_2,\mu_2-1)$&$<\frac{A(\mu_2,\mu_1)}{A(\mu_2,\mu_1+1)}$& \tabincell{c}{$0\le \mu_1<\mu_2\le r$,\\$m<n$}\\
     \hline \tabincell{c}{class--I,\\type-II}&$N(r, \mu_1)$&$N(r, \mu_2)$&$\geq2N(\mu_1+1,\mu_1)$&$\frac{N(\mu_2-1,\mu_1)}{N(\mu_2,\mu_1)}$&$<N(\mu_1+1,\mu_1)$&$<\frac{N(\mu_2,\mu_1)}{N(\mu_2-1,\mu_1)}$&
     \tabincell{c}{$0\le \mu_1<\mu_2\le r$, \\$m<n$}\\
     \hline \tabincell{c}{class--II,\\{$H(\gamma,\rho,q)$}}&{$\gamma q$}&{$\rho q$}&$\geq2\gamma$&{$\frac{1}{\gamma}$}&$<\gamma$&{$<\gamma$}&{\tabincell{c}{$1<\gamma\leq\rho\leq q$,\\$q$ is a prime power}}\\
     \hline
   \end{tabular}
\end{threeparttable}
\end{table*}

\section{Simulations}
\label{simulation}
In this section, we show the empirical performance of the two classes of finite geometry measurement matrices proposed in Section \ref{fgldpc} by several examples.
The proposed matrices have relatively large spark and low coherence, thus their performance can be theoretically guaranteed to some extent under $\ell_0$-minimization and $\ell_1$-minimization, respectively, see Table \ref{tab:parameter}.
Simulation results below show that these matrices perform comparably to, sometimes even better than, the corresponding Gaussian random matrices under both OMP\footnote{Matlab codes can be found at {http://sparselab.stanford.edu/}.} and BP\footnote{Matlab codes can be found at {http://www.cs.ubc.ca/$\sim$mpf/asp/}.}.

If not clearly stated, all simulations are conducted under the following conditions.
The $k$-sparse signals are obtained by firstly selecting the support uniformly at random and then generating nonzero values \emph{i.i.d.} from the standard normal distribution $\mathcal{N}(0, 1)$ (Fig. \ref{fig:pg3401}--\ref{fig:320}, \ref{example:arbmat}) or the Rademacher distribution (Fig. \ref{fig:pt}--\ref{fig:dgfgcompare}, \ref{fig:mixnoise}).
The entries of the Gaussian matrices are chosen \emph{i.i.d.} from $\mathcal{N}(0, 1)$ and then normalized to give fair measurements to all components of the signals.
For each measurement matrix $A$ or $H$ and each signal $\textit{\textbf{x}}$ with sparsity $k$, we conduct an experiment using $M$ Monte Carlo trials ($M=5000$ in Fig. \ref{fig:pg3401}--\ref{fig:320}, \ref{example:arbmat}).
In the $i$-th Monte Carlo trial, a relative recovery error $e_i = ||\textit{\textbf{x}}^*-\textit{\textbf{x}}||_2/||\textit{\textbf{x}}||_2$ is computed, where $\textit{\textbf{x}}^*$ denotes the recovered signal.
If $e_i\leq 0.001$, we declare this recovery to be perfect.
Finally, an average percentage of perfect recovery over the $M$ trials is obtained and shown as a point in the corresponding figures (Fig. \ref{fig:pg3401}--\ref{fig:320}, \ref{fig:dgfgcompare}--\ref{example:arbmat}).

\subsection{Empirical Performance of the class--I Finite Geometry Measurement Matrices}
In this subsection, the class--I finite geometry measurement matrices constructed in Section \ref{subsec:classImat} are used in CS.
\begin{Example}
\label{example:pg3401}
\begin{figure}
\centering
\includegraphics[width=0.4\textwidth]{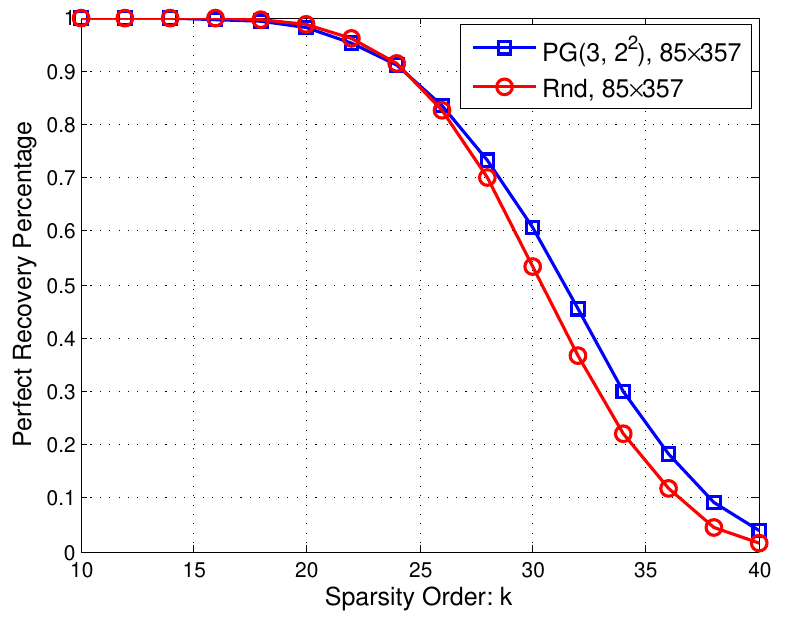}
    \caption {Empirical performance of a type-II projective geometry measurement matrix in $PG(3, 2^2)$ with $\mu_1=0,\;\mu_2=1$ and the corresponding Gaussian matrix under OMP.}\label{fig:pg3401}
\end{figure}
Let $r=3, q=2^2, \mu_{2}=1$, $\mu_{1}=0$,
$PG(3, 2^2)$ consists of $J=357$ lines and $n=85$ points.
Let $H$ be the $J\times n$ line-point incidence matrix in $PG(3, 2^2)$.
By transposing $H$, we can obtain a $(\gamma,\rho)$-regular type-II projective geometry measurement matrix $H^T$, where $\gamma=N_{PG}(1,0)=5$ and $\rho=A_{PG}(1,0)=21$.
Moreover, $H^T$ has {coherence $\mu(H^T) = 1/5$} and ${\rm spark}(H^T)\geq10$.
It is observed from Fig. \ref{fig:pg3401} that $H^T$ has {empirically} similar performance to Gaussian matrix {under OMP}.
For $k<10$, exact recovery is obtained and the corresponding points are not plotted for clear comparisons.
\end{Example}

\begin{figure}
\centering
\includegraphics[width=0.4\textwidth]{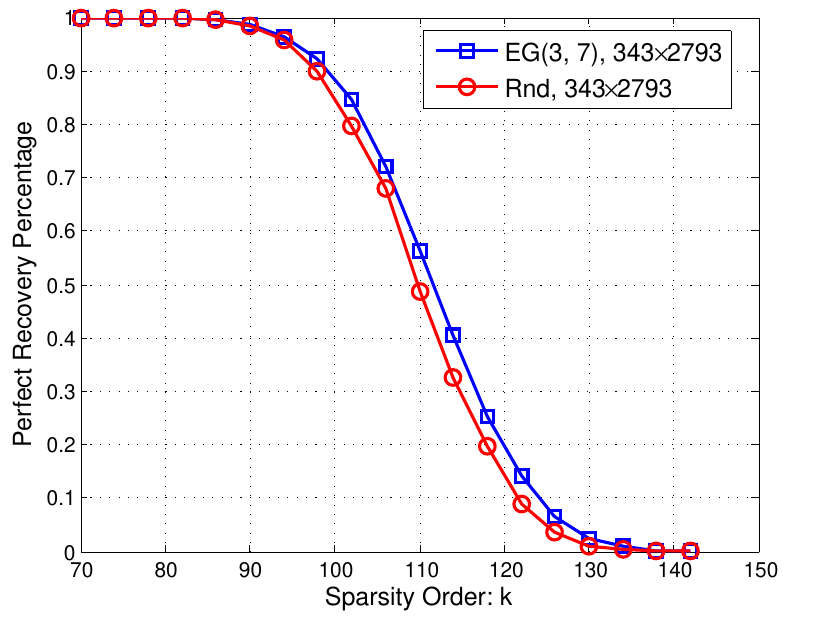}
  \caption {Empirical OMP performance of a type-II Euclidean geometry measurement matrix in $EG(3,7)$ with $\mu_1=0,\;\mu_2=1$ and the corresponding Gaussian matrix.}\label{fig:eg3701}
\end{figure}
\begin{Example}\label{example:eg3701}
Let $r=3$, $q=7$, $\mu_2=1$, and $\mu_1=0$.
Let $H$ be the $2793\times343$ line-point incidence matrix in $EG(3, 7)$. $H^T$ is a $(7,57)$-regular type-II Euclidean geometry measurement matrix with {$\mu(H^T) = 1/7$} and ${\rm spark}(H^T)\geq 14$. Fig. \ref{fig:eg3701} shows the {empirically} good performance of $H^T$ {under OMP}.
\end{Example}

\begin{figure}
\centering
\includegraphics[width=0.4\textwidth]{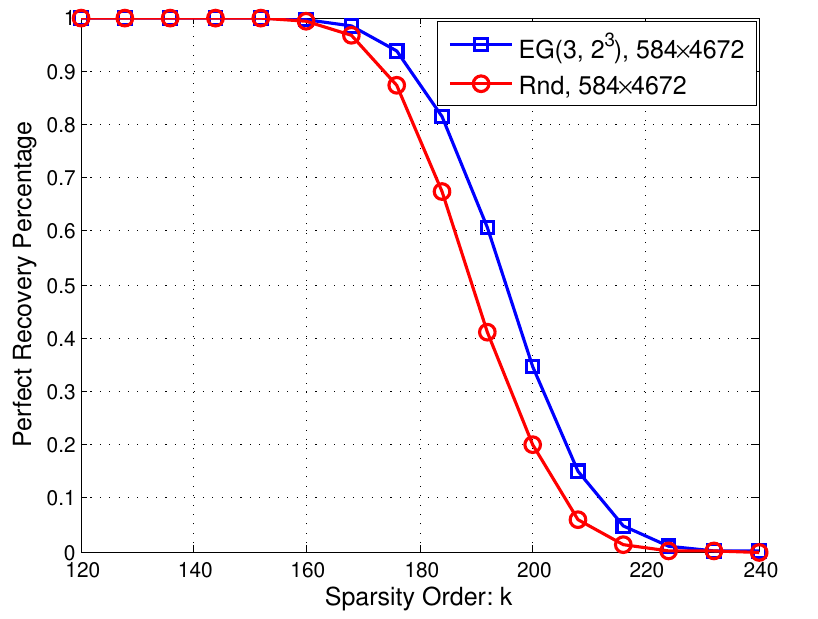}
  \caption {Empirical performance of a type-I Euclidean geometry measurement matrix in $EG(3, 2^3)$ with $\mu_1=1,\;\mu_2=2$ and the corresponding Gaussian matrix under OMP.}\label{fig:eg3812}
\end{figure}
\begin{Example}\label{example:eg3812}
Let $r=3$, $q=2^3$, $\mu_2=2$, $\mu_1=1$ and $H$ be the $(9,72)$-regular type-I incidence matrix in $EG(3, 2^3)$ with size $584\times4672$, {$\mu(H) = 1/9$}, ${\rm spark}(H)\geq 18$.
Fig. \ref{fig:eg3812} shows that
some matrices from finite geometry have {empirically} very good performance {under OMP} for the moderate length signals.
\end{Example}

\subsection{Empirical Performance of the class--II Finite Geometry Measurement Matrices}
In this subsection, the class--II finite geometry measurement matrices constructed in Section \ref{sec:parallel} are used in CS.
\begin{Example}
\label{example:1024}
Let $r=2, q=2^{5}, \mu_{2}=1$, $\mu_{1}=0$, and $H=(H_1^T, H_2^T, \ldots, H_{33}^T)^{T}$ be the $1056\times 1024$ line-point incidence matrix of $EG(2, 2^{5})$.
\begin{figure}
\centering
\includegraphics[width=0.4\textwidth]{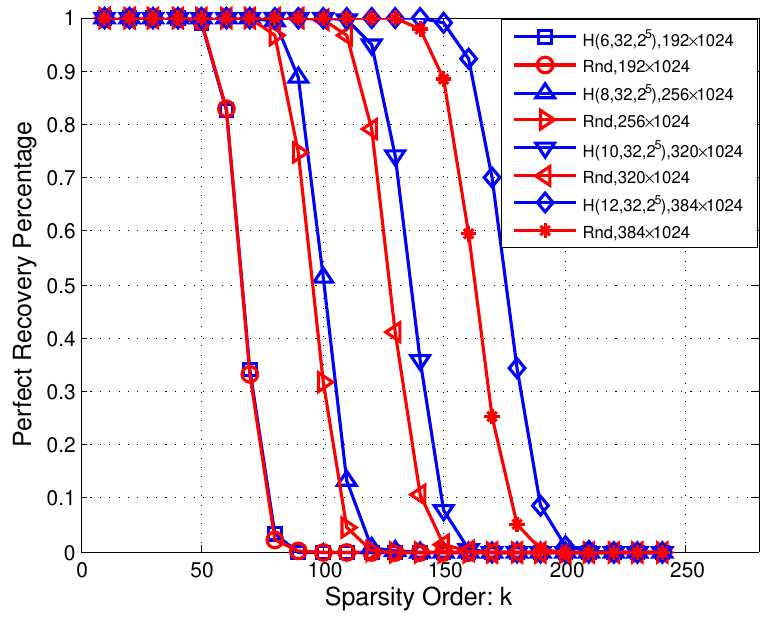}
  \caption {Empirical performance of $H(6,32,2^5)$, $H(8,32,2^5)$, $H(10,32,2^5)$, $H(12,32,2^5)$ and their corresponding Gaussian matrices under OMP.}
  \label{fig:1024}
\end{figure}
Fig. \ref{fig:1024} shows {the OMP recovery} performance of $H(6,32,2^5)$, $H(8,32,2^5)$, $H(10,32,2^5)$, and $H(12,32,2^5)$ with sizes $192\times1024$, $256\times1024$, $320\times1024$ and $384\times1024$, respectively.
{For each $\gamma\in\{6, 8, 10, 12\}$, $\mu(H(\gamma,q,q))= 1/\gamma$, $spark(H(\gamma,q,q))\geq 2\gamma$.}
It is easily observed that all of the submatrices perform better than their corresponding Gaussian matrices, and the more parallel bundles are chosen, the better the submatrix performs, and its gain over Gaussian matrix becomes larger.
\end{Example}
\begin{Example}
\label{example:320}
Consider $H(10,32,2^5)=(H_1^T, \ldots, H_{10}^T)^{T}$ in Example \ref{example:1024}.
By deleting the $(1024-32\rho)$ columns of $H(10,32,2^5)$ corresponding to the points on the first $(32-\rho)$ lines of the 11-th parallel line bundle in $EG(2,2^5)$, we obtain a $(10, \rho)$-regular submatrix $H(10,\rho,2^5)$.
\begin{figure}
\centering
\includegraphics[width=0.4\textwidth]{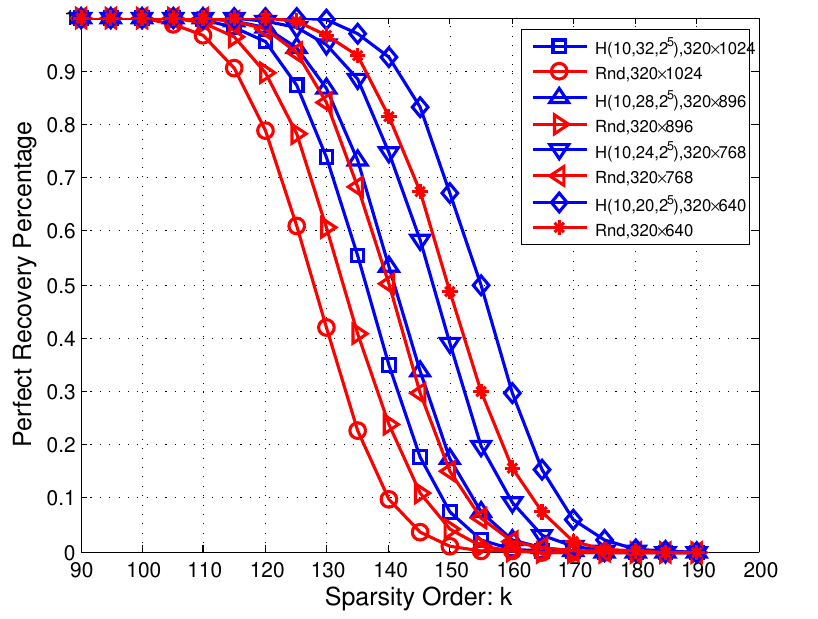}
  \caption {Empirical performance of $H(10,32,2^5)$, $H(10,28,2^5)$, $H(10,24,2^5)$, $H(10,20,2^5)$ and their corresponding Gaussian matrices under OMP.}
  \label{fig:320}
\end{figure}
The 4 blue lines from left to right in Fig. \ref{fig:320} show the performance of $H(10,32,2^5)$, $H(10,28,2^5)$, $H(10,24,2^5)$, and $H(10,20,2^5)$ respectively.
{For any $\rho\in\{20, 24, 28, 32\}$, $\mu(H(10,\rho,2^5))= 1/10$, $spark(H(10,\rho,2^5))\geq 20$.}
Obviously, all of the submatrices perform better than their corresponding Gaussian matrices (the 4 red lines from left to right) {under OMP}.
\end{Example}

In the following, utilizing the \emph{phase transition} curve, we verify that the proposed binary matrices perform empirically as good as Gaussian matrices under BP recovery.
The phase transition phenomenon discovered by Donoho and Tanner \cite{ddtj} says that for any signal $\textit{\textbf{x}}$ with length $n$, sparsity $k=\varepsilon n$ and measured by a Gaussian matrix $A$ with $m=\delta n$ rows, there exists a function $\varepsilon^*(\delta)$ satisfying that for large $n$, if $\varepsilon<\varepsilon^*(\delta)$, BP will recover $\textit{\textbf{x}}$ exactly and fail otherwise with overwhelming probability.
Moreover, \cite{hmsj} show that the phase transition region of many deterministic matrices, such as the DG Frames \cite{rcsh,sjthesis} and Chirp sensing matrices \cite{lasd}, also coincide with those of Gaussian matrices.


\begin{Example}
Let $m$ be 256 or 1024 and in each case, let $n$ vary such that $m/n\in\{1/16, 1/15, \ldots, 1/3, 1/2\}$. Construct 30 binary and regular measurement matrices using the method described in Section \ref{sec:parallel}. In particular, when $m=256$, the matrices come from the line-point incidence matrices of $EG(2,2^5)$ if $m/n\in\{1/2,1/3,1/4\}$ and from $EG(2,2^6)$ otherwise.
When $m=1024$, the matrices are submatrices of the line-point incidence matrices of $EG(2,2^6)$ if $m/n\in\{1/2,1/3,1/4\}$ and from $EG(2,2^7)$ otherwise.
Similar to Fig. S3 in \cite{hmsj}, we can draw the corresponding phase transition curve of the proposed matrices, see Fig. \ref{fig:pt}.
From Fig. \ref{fig:pt}, we can see that the propose matrices in Section \ref{sec:parallel} are empirically as good as Gaussian matrices and other deterministic matrices under BP recovery.

\begin{figure}
\centering
\includegraphics[width=0.4\textwidth]{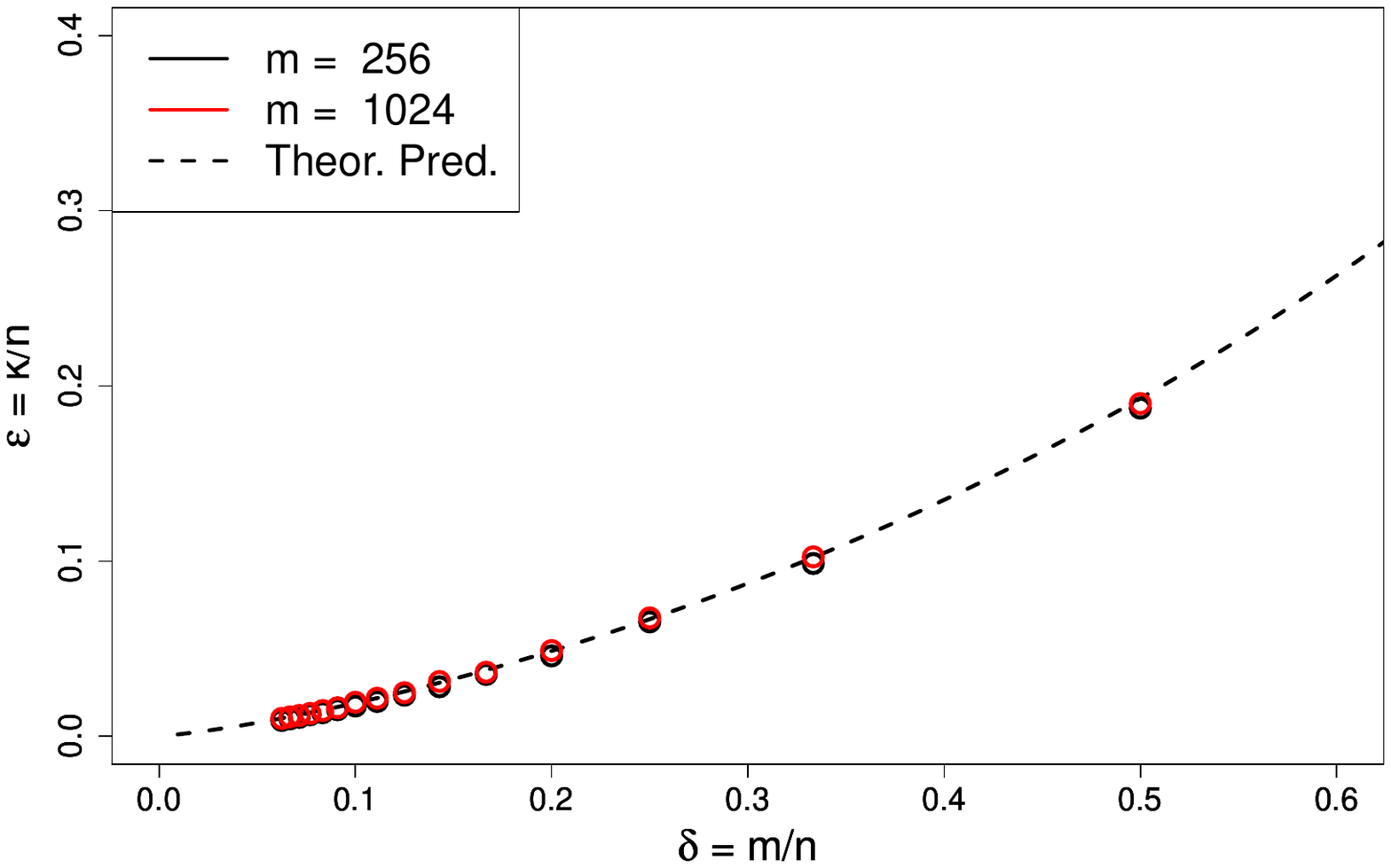}
  \caption {Empirical phase transition curve under BP for the class--II finite geometry measurement matrices. For each $\delta$, the location $\varepsilon$ of 50\% probability of perfect recovery is denoted by (black and red) circles and the asymptotic Gaussian phase transition is shown by the dashed curve.}\label{fig:pt}
\end{figure}
\end{Example}

{
Similar to many other deterministic constructions, the sizes of the proposed matrices are restricted to special numbers.
Sometimes in practice, we often need to construct a deterministic $m\times n$ measurement matrix, where $m$ and $n$ may not happen to be any pair of these special numbers.
Generally, we can obtain the desired matrix by constructing a larger matrix through proper deterministic construction and then removing the extra rows and columns directly.
It is easy to see that the usual theoretical guarantees over a measurement matrix, such as coherence, often keep valid (better or unchanged) if a few columns of the matrix are removed.
As a result, the submatrix obtained by removing extra columns of a deterministic measurement matrix with a special size is expected to perform as well in practice.
However, the influence of removal of rows on the theoretical guarantees seems generally hard to predict, thus resulting in the unpredictability of their empirical performance.}

\begin{Example}\label{example:dgfgcompare}
Construct a real-valued $1024\times 3072$ measurement matrix $A_{\rm DG}$ based on the Kerdock frame $DG(9,0)$ \cite{rcsh,sjthesis}.
Now suppose that we want to construct a $512\times3072$ measurement matrix.
{Let $A_{\rm DG}^{\rm det}$ be a deterministic submatrix of $A_{\rm DG}$ by removing the last $512$ rows of $A_{\rm DG}$ and $A_{\rm DG}^{\rm rnd}$ be a random submatrix of $A_{\rm DG}$ by randomly deleting 512 rows of $A_{\rm DG}$.
See Fig. \ref{fig:dgfgcompare} for the empirical BP recovery performances of $A_{\rm DG}$, $A_{\rm DG}^{\rm det}$, $A_{\rm DG}^{\rm rnd}$ and the corresponding class--II finite geometry matrix $H(16,48,2^6)$, $H(8,48,2^6)$.
\begin{figure}
    \centering
    \includegraphics[width=0.4\textwidth]{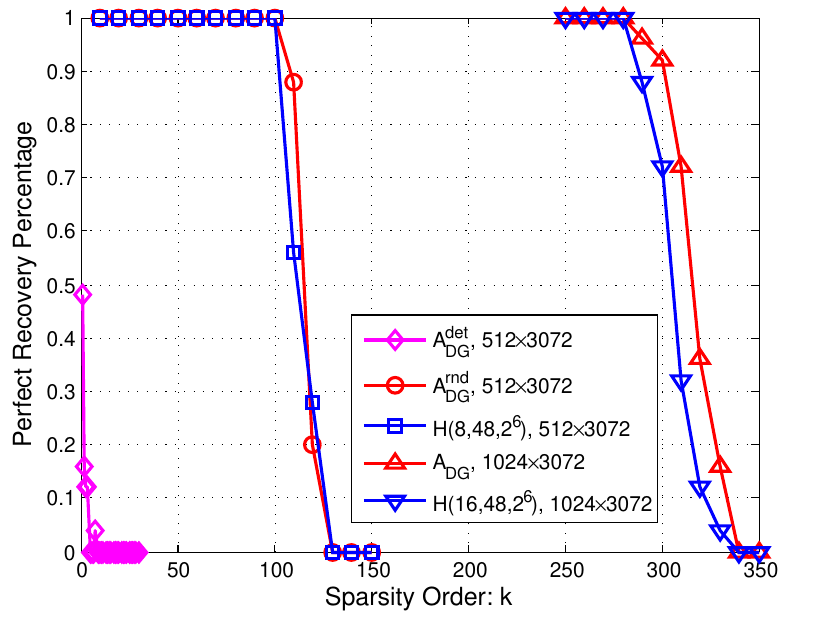}
  \caption {Empirical BP recovery performance of the measurement matrices with size $512\times3072$ (the left 3 curves) and $1024\times3072$ (the right 2 curves) from $DG(9,0)$ and $EG(2,2^6)$. }\label{fig:dgfgcompare}
\end{figure}
In Fig. \ref{fig:dgfgcompare}, $A_{\rm DG}$ and $H(16,48,2^6)$ have similar performance, $A_{\rm DG}^{\rm rnd}$ performs comparably to $H(8,48,2^6)$, but $A_{\rm DG}^{\rm det}$ shows rather bad performance.
Numerical computations show that $\mu(A_{\rm DG})=1/32$, $\mu(H(16,48,2^6)) = 1/16$, $\mu(A_{\rm DG}^{\rm det})=1$, ${\rm spark}(A_{\rm DG}^{\rm det})=2$, $\mu(A_{\rm DG}^{\rm rnd})=0.1289$, $\mu(H(8,48,2^6))=1/8$ while ${\rm spark}(H(16,48,2^6)) \ge 32$, ${\rm spark}(H(8,48,2^6))\geq 16$ according to Theorem \ref{preth}.
Additionally, we generate $A_{\rm DG}^{\rm rnd}$ for several times and calculate their coherence by Matlab, see Table \ref{tab:coh} for five of the results.
\begin{table}[htbp]
\renewcommand{\arraystretch}{1.2}
   \caption{The Coherence $\mu$ of Submatrices Obtained by Deterministically or Randomly Choosing 512 Rows from $A_{\rm DG}$.}
   \label{tab:coh}
   \begin{tabular}{|c|c|c|c|c|c|c|}
     \hline matrix&$A_{\rm DG}^{\rm det}$&$A_{\rm DG}^{\rm rnd}$-1&$A_{\rm DG}^{\rm rnd}$-2&$A_{\rm DG}^{\rm rnd}$-3&$A_{\rm DG}^{\rm rnd}$-4&$A_{\rm DG}^{\rm rnd}$-5\\
     \hline $\mu$&1&0.1406&0.1680&0.1953&0.1328&0.1211\\
     \hline
   \end{tabular}
\end{table}
From Fig. \ref{fig:dgfgcompare} and Table \ref{tab:coh}, we can see that the coherence and empirical performance of the submatrices obtained by removing some rows from $A_{\rm DG}$ seem to be unstable in practice.}
\end{Example}

In the following example, we will show that the submatrices obtained by removing some extra rows and columns from the class--II finite geometry matrices often have relatively stable empirical performance.
Afterwards, an intuitive explanation for this phenomenon by the coherence and spark will be given.

\begin{Example}\label{example:arbmat}
Suppose $m_1=270$, $m_2=285$, $m_3=300$, $n_1=820$, $n_2=840$, $n_3=860$.
By removing the last extra rows and columns of $H(9,26,2^5)$, $H(9,27,2^5)$ and $H(10,27,2^5)$, we can construct three binary matrices $H_{\rm sub}(9,26,2^5)$, $H_{\rm sub}(9,27,2^5)$ and $H_{\rm sub}(10,27,2^5)$ with sizes $m_1\times n_1$, $m_2\times n_2$ and $m_3\times n_3$, respectively.
Their performances are shown in Fig. \ref{fig:arbmat}, which indicates that all of them perform empirically better than the corresponding Gaussian matrices under OMP.
\begin{figure}
\centering
\includegraphics[width=0.4\textwidth]{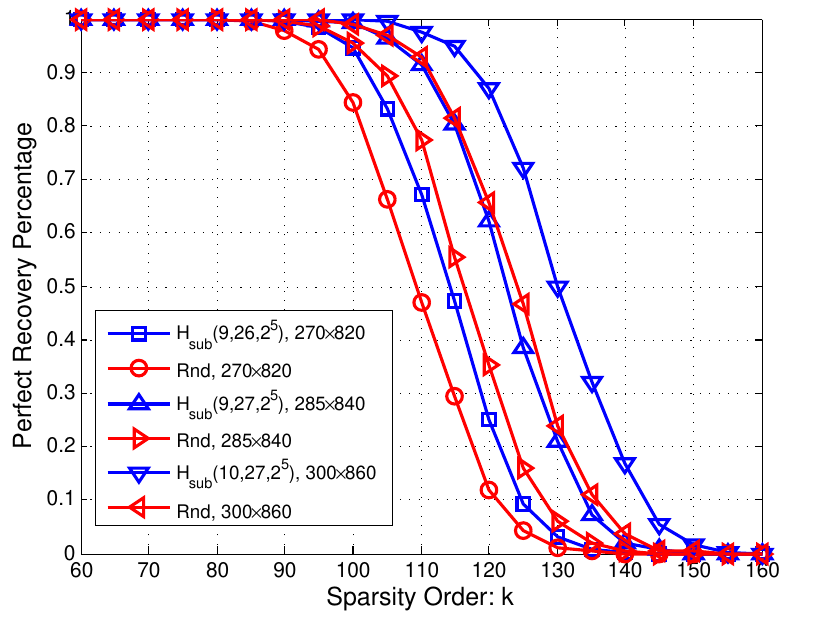}
  \caption {Empirical performance of the submatrices of $H(9,26,2^5)$, $H(9,27,2^5)$ and $H(10,27,2^5)$ with sizes $270\times820$, $285\times 840$ and $300\times 860$ and their corresponding Gaussian matrices under OMP.}
  \label{fig:arbmat}
\end{figure}
Numerical computations by Matlab show that $\mu(H_{\rm sub}(9,26,2^5))=\mu(H_{\rm sub}(9,27,2^5))= 1/8$, $\mu(H_{\rm sub}(10,27,2^5))=1/9$ and according to Theorem \ref{preth}, ${\rm spark}(H_{\rm sub}(9,26,2^5))\geq16$, ${\rm spark}(H_{\rm sub}(9,27,2^5))\geq16$ and ${\rm spark}(H_{\rm sub}(10,27,2^5))\geq 18$.
\end{Example}
{
\begin{Remark}\label{rem:coherenceanalysis}
If we remove any $0\le c_1<q$ rows among the last $q$ rows and any $0\le c_2<q$ columns from $H(\gamma,\rho,q)$, where $1<\gamma\le \rho\le q$, the coherence of the resulting submatrix $H_{\rm sub}(\gamma,\rho,q)$ will satisfy
    $$\frac{1}{\gamma}\leq\mu(H_{\rm sub}(\gamma,\rho,q))\leq\frac{1}{\gamma-1},$$
since $H_{\rm sub}(\gamma,\rho,q)$ has either column weight $\gamma-1$ or $\gamma$ and maximum inner product 1 between any two columns when $\rho\ge2$.
In addition,  according to Theorem \ref{preth}:
$${\rm spark}(H_{\rm sub}(\gamma,\rho,q))\geq 2(\gamma-1).$$
Fig. \ref{fig:cohchange} illustrates how the lower bounds of $1/\mu$ and spark of $H_{\rm sub}(\gamma,\rho,q)$ change when $m$ increases from $32$ to $1024$ for $q=2^5$.
  \begin{figure}
    \centering
    \includegraphics[width=0.4\textwidth]{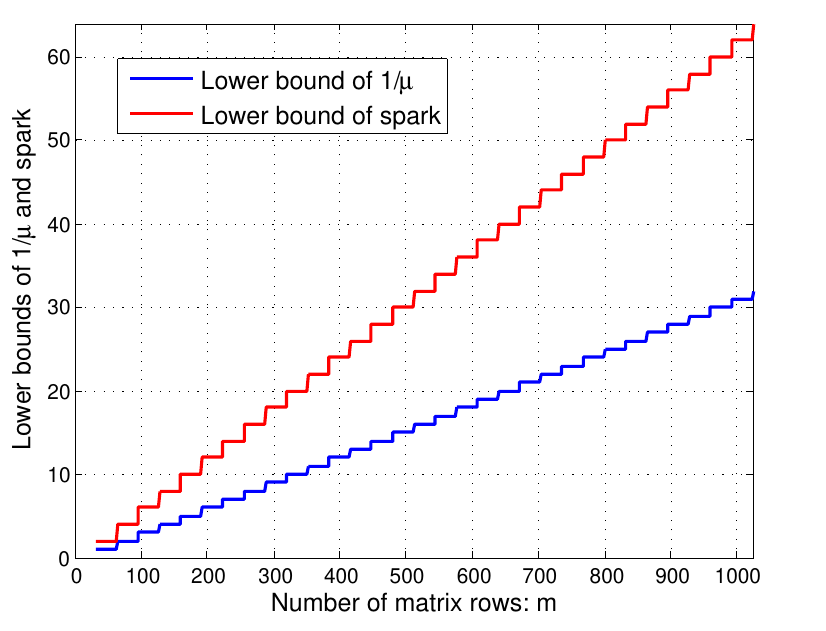}
    \caption {The variation trends of the lower bounds of $1/\mu$ and spark with respect to $m$ for the $m\times n$ measurement matrix $H_{\rm sub}(\gamma,\rho,q)$, where $(\gamma-1)q<m\le\gamma q$, $(\rho-1)q<n\le \rho q$, $1\le\gamma\le \rho\le q$, and $q=2^5$.}
  \label{fig:cohchange}
\end{figure}
In Fig. \ref{fig:cohchange}, the lower bounds of spark and $1/\mu$ grow \emph{gradually}, rather than sharply or unstably, when $m$ increases, thus preventing the performance of $H(\gamma,\rho,q)$ from deteriorating too fast when some rows are removed, which can explain to some extent the empirically good performance of $H_{\rm sub}(\gamma,\rho,q)$.
\end{Remark}
}

\subsection{Stable and Robust Compressed Sensing Under Measurement Matrices from Finite Geometry}
In practice, the signals are often approximately sparse instead of exactly sparse (namely, they have data-domain noises) and the measurements can also be corrupted by measurement-domain noises.
In the following, we will explore the {empirical} performance of the proposed matrices under noises.

{Consider the practical compressed sensing problem:
\begin{equation}
  \textit{\textbf{y}}=A(\textit{\textbf{x}}+\textit{\textbf{e}}_d)+\textit{\textbf{e}}_m,
\end{equation}
where $\textit{\textbf{e}}_d\in\mathbb{R}^n$ and $\textit{\textbf{e}}_m\in\mathbb{R}^m$ stand for the data-domain and measurement-domain noise, respectively.
Model $\textit{\textbf{e}}_d$ and $\textit{\textbf{e}}_m$ as the Gaussian vectors with each entry \emph{i.i.d.} chosen from the Gaussian distribution $\mathcal{N}(0, \delta_d)$ and $\mathcal{N}(0, \delta_m)$, respectively.}

\begin{Example}
Let $m=1024$, $n=3072$, $k=200$ and normalize the sparse signals.
Construct an $m\times n$ binary matrix from the line-point incidence matrix of $EG(2,2^6)$.
As comparisons, we also construct the Gaussian matrix and real-valued frame $DG(9,0)$ \cite{rcsh}\cite{sjthesis} with the same size, which are representatives for random and deterministic measurement matrices, respectively.
Change $\delta_d$ and $\delta_m$ independently from $10^{-6}$ to $10^{-1}$ and plot the average {BP\footnote{Actually, the BP denoising algorithm is used here to deal with the noise.}} recovery $SNR =-10\log(||\textit{\textbf{x}}^*-\textit{\textbf{x}}||_2/||\textit{\textbf{x}}||_2)$
as a function of $\delta_m$ and $\delta_d$, see Fig. \ref{fig:mixnoise}.
\begin{figure*}
\centering
\subfigure[]{\includegraphics[width=0.32\textwidth]{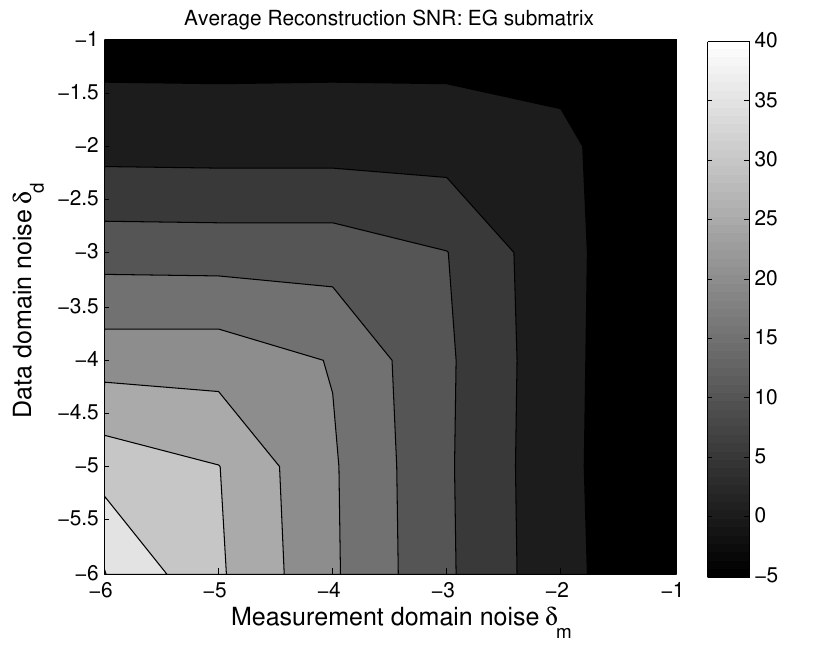}}
\subfigure[]{\includegraphics[width=0.32\textwidth]{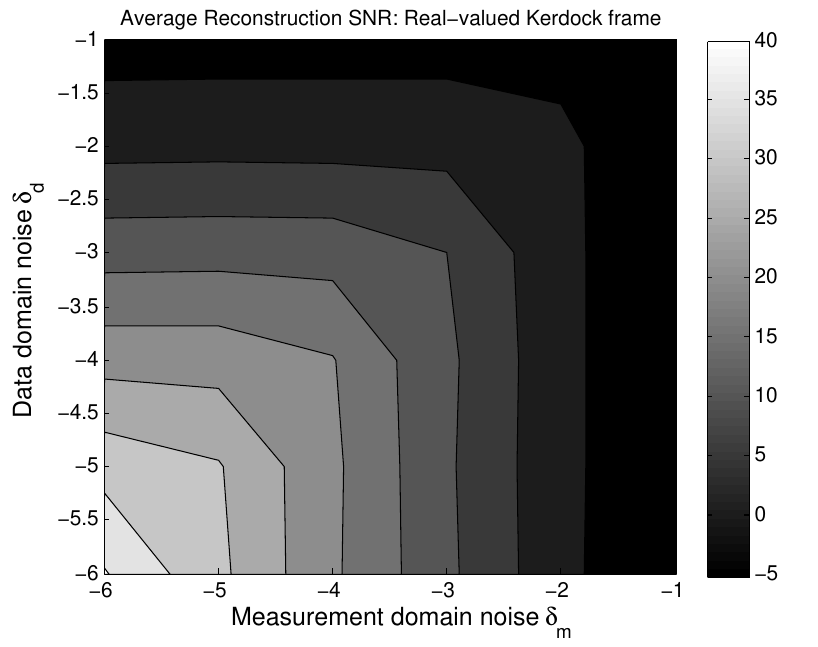}}
\subfigure[]{\includegraphics[width=0.32\textwidth]{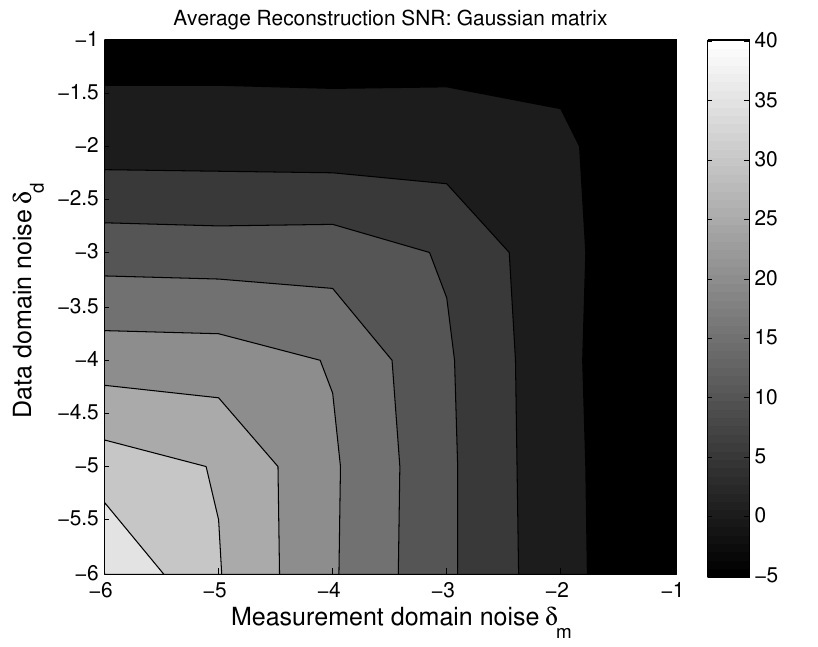}}
\caption {Average BP recovery SNR as a function of the data-domain noise $\delta_d$ and measurement-domain noise $\delta_m$ for (a) binary measurement matrix from $EG(2,2^6)$, (b) real-valued $DG(9,0)$ frame, and (c) random Gaussian matrix.
$m=1024$, $n=3072$ and $k=200$.}\label{fig:mixnoise}
\end{figure*}
It is easily seen that similar to Gaussian matrices and real-valued DG frames, the proposed matrices have stable and robust {empirical} performance in practice.
\end{Example}
%
\begin{Example}
  At last, we apply the binary matrix from finite geometry to the Lena image with size $128\times128$.
\begin{figure}
\centering
\subfigure[Original image] {\includegraphics[width=0.2\textwidth]{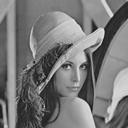}}
\subfigure[25\% sparse image ]{\includegraphics[width=0.2\textwidth]{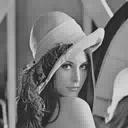}}
\subfigure[FG, PSNR = 26.8 ]{\includegraphics[width=0.2\textwidth]{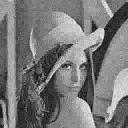}}
\subfigure[Rnd, PSNR = 24.9 ]{\includegraphics[width=0.2\textwidth]{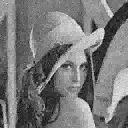}}
\caption {Comparison of the OMP recovery performance by using binary sampling matrix with size $8192\times 16384$ from finite geometry and the corresponding Gaussian matrix to compress a $128\times128$ Lena image with sparsity of 25\%.}\label{fig:image}
\end{figure}
In Fig.\ref{fig:image}, we firstly sparsify the image by discarding 75\% of its smallest Haar wavelet coefficients.
Afterwards, the binary measurement matrix with size $8192\times 16384$, which is a submatrix of the line-point incidence matrix of $EG(2,2^7)$ and the corresponding Gaussian matrix are employed to measure the sparsified image.
Finally, the image is reconstructed by the OMP algorithm and it is easily seen that the proposed binary matrix outperforms the Gaussian matrix by about 1.9dB.
\end{Example}
\section{Conclusions and Discussions}
\label{conclusion}
In this paper, by drawing methods and results from LDPC codes, we study the deterministic constructions and performance evaluation of binary measurement matrices.
Lower bounds of spark were proposed for real matrices in \cite{ddme} many years ago. When the real matrices are changed to binary matrices, better results emerge. Firstly, a lower bound of spark is obtained for general binary matrices, which improves the one derived from \cite{ddme} in most cases. Then, we propose two classes of deterministic binary measurement matrices based on finite geometry. One class is the incidence matrix $H$ of $\mu_2$-flat over $\mu_1$-flat in finite geometry $FG(r,q)$ or its transpose $H^T$. The other class is the submatrix of $H$ or $H^T$, especially the matrix obtained by deleting row parallel bundles or column parallel bundles from $H$ or $H^T$ in Euclidean geometry.
Many of the proposed matrices have cyclic or quasi-cyclic structure \cite{txla05} which make the hardware realization convenient and easy.
For the class--I finite geometry measurement matrix, two further improved lower bounds of spark are given to show their relatively large spark.
Finally, lots of simulations, including the noiseless and noisy situations, are done according to standard and comparable procedures. Simulation results show that {empirically}, the proposed matrices perform comparably to, and sometimes even better than the corresponding Gaussian random matrices.

Spark is a necessary condition to guarantee practical recovery (such as $\ell_1$-minimization) performance, thus it is often weak and lacks stability in practice.
However, as (\ref{kl1new}) has indicated, sometimes the bounds derived by spark and $\ell_0$-minimization, e.g. Theorem \ref{preth}, also agree with the corresponding results in $\ell_1$-minimization.
Therefore, it seems interesting to investigate whether the rest results about the lower bounds of sparks can be extended to $\ell_1$-minimization.

{
Finally, we discuss a bit more about the famous open problem posed by Tao in 2007 \footnote{http://terrytao.wordpress.com/2007/07/02/open-question-deterministic-uup-matrices/}: constructing deterministic matrices satisfying RIP of order $k=O(m/\log(n/m))$.
As has been indicated in the introduction, most of the existing deterministic constructions are based on coherence and thus can only be shown to hold RIP of order $k=O(\sqrt{m})$.
Up to now, only one remarkable breakthrough on this open problem was made.
In \cite{jbsd}, leveraging the additive combinatorics, Bourgain, \emph{et al.} were able to construct a deterministic RIP matrix with $k=O(m^{1/2+\epsilon_0})$, where $\epsilon_0\approx 5.5169\times 10^{-28}$ and this breaks the notorious square-root bottleneck $k=O(\sqrt{m})$.
Recently, Mixon has made some further progresses and increased $\epsilon_0$ to  $\epsilon_0\approx4.4466\times 10^{-24}$, see \cite{dustin, wikidustin,hpdustin} for more details.
Besides these progresses concerning RIP, there are also some other related contributions.
For example, using the PEG algorithm for constructing the parity-check matrices of LDPC codes, Tehrani \emph{et al.} proposed a family of explicit NSP matrices with $m=O(k\log(n/k))$ that can recover \emph{most} $k$-sparse signals under $\ell_1$-minimization with probability $1-1/n$ \cite{akat,atad}.
In \cite{xlsx2}, we showed that for a binary matrix $H$, ${\rm spark}(H)\geq d(\mathcal{C})$,
where $d(\mathcal{C})$ denotes the minimum distance of the binary code $\mathcal{C}$ defined by $H$.
Thus, we can construct an explicit $m\times n$ matrix $H$ with $m=O(n)$ by using an LDPC code $\mathcal{C}$ with $d(\mathcal{C})=O(n)$ (e.g. \cite{tras,ddsd}), such that \emph{any} $k=O(n)$-sparse signal measured by $H$ can be perfectly reconstructed under $l_0$-optimization.
As a result, perhaps certain special parity-check matrices of LDPC codes could be promising to help solve this open problem.}


%

{
\appendix
In this appendix, we give an example of measurement matrix with large spark but poor nullspace property (NSP). Firstly, let us look at the definition of (strict) nullspace property.
\begin{Definition}
{\rm \cite{wxbh,adrs}}
  Let $A\in \mathbb{R}^{m\times n}$, $k\in \mathbb{N}^{*}$, and $C\geq1$.
  We say that $A$ has the \emph{strict nullspace property} $NSP_{\mathbb{R}}^{<}(k,C)$, i.e., $A\in NSP_{\mathbb{R}}^{<}(k,C)$, if $\forall K\in[n]$ with $|K|\leq k$,
  \begin{equation*}
    C\cdot||\textit{\textbf{w}}_K||_1<||\textit{\textbf{w}}_{\bar{K}}||_1,\quad \forall \textit{\textbf{w}}\in Nullsp_{\mathbb{R}}^*(A).
  \end{equation*}
\end{Definition}

It is well known that $l_{1}$-optimization can exactly recover any $k$-sparse signal if and only if $A\in NSP_{\mathbb{R}}^{<}(k,C=1)$ \cite[Th.1]{wxbh}\cite[Th.3]{adrs}.

Consider the $m\times n$ Vandermonde matrix with $m\le n$:
\begin{eqnarray*}
A(m,n) = \left(
\begin{array}{ccccc}
1&1&1&\cdots&1 \\
\alpha_1&\alpha_2&\alpha_3&\cdots&\alpha_n \\
\vdots&\vdots&\vdots&\ddots&\vdots \\
\alpha_1^{m-1}&\alpha_2^{m-1}&\alpha_3^{m-1}&\cdots&\alpha_n^{m-1}
\end{array}
\right).
\end{eqnarray*}
Suppose $\alpha_1, \alpha_2, \alpha_3, \cdots,\alpha_n$ are different from each other, then $A(m,n)$ is full spark \cite{bajc}, i.e., ${\rm spark}(A(m,n))=m+1$.
Let $\textit{\textbf{w}}\in \mathbb{R}^{n}$ be a column vector with its $i$-th entry to be $$\frac{1}{\Pi_{i\neq j} (\alpha_{i}-\alpha_{j})},\quad 1\le i,j\le n.$$
It is easy to check that $\textit{\textbf{w}}\in{\rm Nullsp}_{\mathbb{R}}^{*}(A(m,n))$ \footnote{http://mathoverflow.net/questions/49255/how-to-determine-the-kernel-of-a-vandermonde-matrix.} when $m<n$.

Let $n=21$, $m=n-1$ and $\alpha_i = i$, $1\le i\le 21$, then the Vandermonde matrix $A(20,21)$ has ${\rm spark}(A(20,21))=21$.
However, numerical computations by Matlab show that $|w_{11}|>|w_{10}|=|w_{12}|>|w_9|=|w_{13}|>|w_{j}|$ for any $1\le j\le 8$ and $14\le j\le 21$, $|w_{10}|+|w_{11}|+|w_{12}|<\frac{||\textit{\textbf{w}}||_1}{2}$ and  $|w_{9}|+|w_{10}|+|w_{11}|+|w_{12}|>\frac{||\textit{\textbf{w}}||_1}{2}$, which implies $A(20,21)\in NSP_{\mathbb{R}}^{<}(k_{\ell_1}=3,C=1)$.
By (\ref{proeq1}), any signal with sparsity $k_{\ell_0}\le10$ can be recovered by $\ell_0$-minimization, while from NSP, all signals with sparsity $k_{\ell_1}\le3$ can obtain perfect reconstruction by $\ell_1$-minimization.
}

\section*{Acknowledgment}
The authors would like to thank Mr. Hatef Monajemi for explaining some details of \cite{hmsj} patiently by emails.
{The authors wish to express their sincere gratefulness to the two anonymous reviewers and the associate editor, Prof. Akbar Sayeed, for their valuable suggestions and comments that helped to greatly improve this paper.}
\ifCLASSOPTIONcaptionsoff
  \newpage
\fi

\end{document}